\documentclass[conference,letterpaper]{IEEEtran}

\usepackage[utf8]{inputenc} 
\usepackage[T1]{fontenc}
\usepackage{url}
\usepackage{ifthen}
\usepackage{cite}
\usepackage[cmex10]{amsmath}
\usepackage{amsthm, latexsym, amssymb, amsfonts, epsf, xcolor, mathrsfs,graphicx}
\usepackage{algorithm}
\usepackage{algpseudocode}

\newcommand{\F}{\mathbb{F}}

\newcommand{\wt}{\textnormal{wt}}
\newcommand{\row}{\text{Row}}

\DeclareMathOperator{\hull}{Hull}
\DeclareMathOperator{\supp}{Supp}

\newtheorem{theorem}{Theorem}
\newtheorem{proposition}[theorem]{Proposition}
\newtheorem{lemma}[theorem]{Lemma}
\newtheorem{corollary}[theorem]{Corollary}
\newtheorem{definition}[theorem]{Definition}
\newtheorem{remark}[theorem]{Remark}
\newtheorem{example}[theorem]{Example}

\newcommand{\rmv}[1]{}

\IEEEoverridecommandlockouts

\begin{document}

\title{Toward Quantum CSS-T Codes from Sparse Matrices}

\author{
\IEEEauthorblockN{Eduardo Camps-Moreno, \thanks{Camps-Moreno and Matthews are partially supported by NSF DMS-2201075. L\'{o}pez is partially supported by NSF DMS-2401558. McMillon is supported by NSF DMS-2303380.  The authors are partially supported by the Commonwealth Cyber Initiative.} Hiram H.\ L\'{o}pez, Gretchen L. Matthews, \& Emily McMillon}  \IEEEauthorblockA{Department of Mathematics\\
                    Virginia Tech\\
                    Blacksburg, VA, USA\\
                    Email: \{e.camps, hhlopez, gmatthews, emcmillon\}@vt.edu}
}

\maketitle

\begin{abstract}
CSS-T codes were recently introduced as quantum error-correcting codes that respect a transversal gate. A CSS-T code depends on a pair $(C_1, C_2)$ of binary linear codes $C_1$ and $C_2$ that satisfy certain conditions. We prove that $C_1$ and $C_2$ form a CSS-T pair if and only if $C_2 \subset \hull(C_1) \cap \hull(C_1^2)$, where the hull of a code is the intersection of the code with its dual. We show that if $(C_1,C_2)$ is a CSS-T pair, and the code $C_2$ is degenerated on $\{i\}$, meaning that the 
$i^{th}$-entry is zero for all the elements in $C_2$, then the pair of punctured codes $(C_1|_i,C_2|_i)$ is also a CSS-T pair. Finally, we provide Magma code based on our results and quasi-cyclic codes as a step toward finding quantum LDPC or LDGM CSS-T codes computationally.
\end{abstract}



\section{Introduction}
\label{section:intro}

Since the 1990s, it has been known that linear codes $C_1$ and $C_2$ with $C_2 \subseteq C_1$ may be used to define a  quantum stabilizer code $Q(C_1,C_2)$ called a CSS code, based on the work of Calderbank and Shor \cite{CalderbankShor_96} and Steane \cite{S96}. Stabilizer codes are important for quantum error correction, protecting quantum information against errors induced by noise and decoherence, and fault-tolerant quantum computing. It is desirable that quantum error-correcting codes enable gates to be implemented transversally, since transversal gates act independently on the qubits to impede the propagation of errors. Codes with this property are called CSS-T codes, as introduced in \cite{calderbankclassicalcsst}. More formally, 
the code $Q(C_1,C_2)$ is a CSS-T code provided $C_2$ is an even weight code, meaning all of its codewords have even weight, and for each codeword $c \in C_2$, the shortening of $C_1^{\perp}$ with respect to the support of $c$ is self-dual. Reed-Muller codes have been employed to define CSS-T codes \cite{felicecsst}, and the mathematical concepts defining binary CSS-T codes have been considered for larger alphabets \cite{albertocsst}. In \cite{campsmoreno2023poset}, it was demonstrated that CSS-T codes form a poset, providing tools to study maximal CSS-T codes. Even so, CSS-T codes (especially those satisfying other desirable properties) are elusive and determining asymptotically good families of CSS-T codes is an open question. 

Another important family of quantum error-correcting codes is formed by the quantum low-density parity-check (LDPC) codes \cite{Breuckmann}, \cite{G} which are CSS codes based on sparse parity-check matrices. In recent work, such codes have been demonstrated for constant-overhead fault-tolerant quantum computation  \cite{xu2023constantoverhead} and to
rival surface codes  \cite{IBM}, which were introduced in \cite{bravyi1998quantum}.
Additional evidence of the capabilities of quantum LDPC codes may be found in  \cite{TZ_14},
\cite{Panteleev_22}, \cite{LZ_22}, \cite{EKZ_22}, \cite{krishna2023vidermans}.

In this paper, we consider  CSS-T codes defined by sparse matrices. We demonstrate a step towards finding a source of such codes by providing an equivalent characterization of CSS-T codes and considering a construction of LDPC codes that have efficient encoding \cite{LGC_19}. This paper is organized as follows. Section \ref{section:prelim} reviews the necessary background. Section \ref{section:hulls} includes a useful characterization of binary CSS-T codes in terms of hulls and relative hulls of codes. In Section \ref{section:families}, we consider the application to sparse matrix codes. Examples are provided in Section \ref{section:examples}. A conclusion is given in Section \ref{section:conclusion}.

\section{Preliminaries}
\label{section:prelim}

In this section, we review the foundation on which our results will be built and set the notation to be used later. We use the standard notation from coding theory. All codes considered in this paper are binary linear codes. The finite field with two elements is $\F_2:=\{ 0, 1 \}$, and $\mathcal A^{m \times n}$ is the set of all $m \times n$ matrices whose entries are elements of a set $\mathcal A$. Given a matrix $A \in \F_2^{m \times n}$, its transpose is $A^T \in \F_2^{n \times m}$ and its $i^{\text{th}}$ row is $\row_iA$. The standard basis vectors for $\F_2^n$ are $e_i=(0, \dots, 0, 1, 0, \dots, 0)$, $i \in [n]:=\left\{ 1, \dots, n \right\}$, where the only nonzero entry is in position $i$, and $\mathbf{1}=\sum_{i \in [n]} e_i \in \F_2^n$ denotes the all ones vector. An $[n,k,d]$ code $C$ is a $k$-dimensional vector subspace of $\F_2^n$ in which any two distinct codewords (meaning elements of $C$) differ in at least $d$ positions.

The Schur (also known as the star or pointwise) product of two vectors $x=(x_1,\ldots,x_n), y=(y_1,\ldots,y_n) \in F_2^n$ is given by
$$x \star y := (x_1y_1,\ldots, x_ny_n).$$ The Schur (also called the star or pointwise) product of binary codes $C$ and $C'$ of length $n$ is denoted and defined by
$$
C \star C':=\left< c \star c': c \in C, c' \in C'\right> \subseteq \F_2^n,
$$
the $\F_2$-span of the Schur products of all codewords of $C$.
We define the square of the code $C$ by $C^2:=C \star C$.

\subsection{Codes defined by sparse matrices}
A binary linear code is a low-density parity-check (LDPC) code if it is the null space of a matrix $H \in \F_2^{r \times n}$ with row weight $O(1)$. LDPC codes have superior performance when coupled with iterative message passing algorithms that operate on an associated sparse Tanner graph.
A binary linear code is a low-density generator matrix (LDGM) code if it has a generator matrix $G \in \F_2^{k \times n}$ with row weight $O(1)$. LDGM codes facilitate efficient encoding. One may also consider moderate-density parity-check (MDPC) or moderate-density generator matrix (MDGM) codes, which have defining matrices with row weight $O(\sqrt{n \log n})$ or $O(\sqrt{n})$. For LDPC codes, row weights of less than 10 are typically considered \cite{Misoczki_MDPC}.

\subsection{CSS codes}
We say that $Q$ is $[[n,k,d]]$ to mean a quantum code that encodes $k$ logical qudits into $n$ physical qudits and minimum distance $d$. Let $C_1 = [n, k_1, d_1]$ and $C_2 = [n, k_2, d_2]$ be two binary codes such that $C_2 \subseteq C_1$. We denote by $Q(C_1,C_2)$ the quantum CSS code defined by $C_1$ and $C_2$. Recall that a binary generator matrix for the stabilizer that defines $Q(C_1,C_2)$ can be written as
\begin{equation} \label{eq:css}
\left[ 
\begin{array}{cc}
0 & H_1 \\ 
G_2 & 0
\end{array}
\right],
\end{equation}
where $H_1$ is a parity-check matrix for $C_1$ and $G_2$ is a generator matrix for $C_2$. The CSS code $Q(C_1,C_2)$ is a
$$[[n,k_1+k_2-n,\ge \min\{d_1,d_2\}]]$$ quantum code. Notice that the matrix in Expression (\ref{eq:css}) will be sparse if $H_1$ is the parity-check matrix of an LDPC code and $G_2$ is the generator matrix of an LDGM code. 

\subsection{Quasi-cyclic codes}
We will use quasi-cyclic codes to define the quantum codes of interest.
 Recall that a code $C \subseteq \F_2^n$ is quasi-cyclic if and only if there exists $l \in [n]$ such that $c=(c_1, \dots, c_n) \in C$ implies $(c_{n-l+1}, \dots, c_n, c_1, \dots, c_{n-l}) \in C$, where indices are taken modulo $n$. Note that the code $C$ is cyclic if it is quasi-cyclic for $l=1$. Quasi-cyclic codes can be defined by circulant matrices, meaning those of the form
$$
\left[ 
\begin{array}{cccc}
a_1 & a_2 & \dots & a_n \\
a_n & a_1 & \dots & a_{n-1}\\
\vdots & \vdots & & \vdots \\
a_2 & a_3 & \dots & a_1
\end{array}
\right],
$$
 used as blocks $A_{ij}$ in a generator or parity-check matrix of the form
 $$
\left[ 
\begin{array}{cccc}
A_{11} & A_{12} & \dots & A_{1n} \\
A_{21} & A_{22} & \dots & A_{2n}\\
\vdots & \vdots & & \vdots \\
A_{t1} & A_{t2} & \dots & A_{tn}
\end{array}
\right].
$$ 
Quasi-cyclic codes can be compactly described, and, as the next result shows, retain their structure under the Schur product.

\begin{proposition}
    If $C$ is quasi-cyclic, then $C^2$ is quasi-cyclic.
\end{proposition}
    
\begin{proof}
Suppose $C \in \F_2^n$ is quasi-cyclic and $w \in C^2$. Then $w=\sum_{c,c' \in C} a_{c,c'} c \star c'=\sum_{c,c' \in C} a_{c,c'} (c_1 c_1', \dots, c_n c_n') $ where $a_{c,c'} \in \F_2$. Notice that  $c, c' \in C$ implies 
$(c_{n-l+1}c_{n-l+1}', \dots, c_nc_n', c_1c_1', \dots, c_{n-l}c_{n-l}') \in C$
since $(c_{n-l+1}, \dots, c_n, c_1, \dots, c_{n-l}) \in C$ and $(c_{n-l+1}', \dots, c_n', c_1', \dots, c_{n-l}')  \in C$ which completes the proof.
\end{proof}

\section{CSS-T Codes and Hulls} \label{section:hulls}

Hulls were first considered by Assmus and Key \cite{assmus} and have been considered in the context of linearly complementary dual (LCD) codes \cite{MASSEY1992337} as well as entanglement-assisted quantum error-correcting codes \cite{LGC_19},  \cite{anderson2023relative}.  In this section, we consider a characterization of CSS-T codes and link it to the hull of a code. First, we recall the definition of a CSS-T code as given in \cite{calderbankclassicalcsst}. We say that a code $C$ is of even weight if and only if all of its codewords are of even weight.

\begin{definition} \label{def:CSST}
     Given binary codes $C_1$ and $C_2$ of length $n$ with $C_2 \subseteq C_1$, $Q(C_1, C_2)$ is a CSS-T code if and only if $C_2$ is of even weight and for all codewords $c\in C_2$, the shortening of $C_1^{\perp}$ with respect to the coordinates of $[n] \setminus supp(c)$ contains a self-dual code.
\end{definition}

We use the following result from \cite{campsmoreno2023poset} to determine quantum  codes which are CSS-T. 

\begin{theorem}\label{T:equiv-def} Let $C_1$ and $C_2$ be binary codes of length $n$.
    The following are equivalent.
    \begin{itemize}
        \item[\rm (1)] $Q(C_1, C_2)$ is a CSS-T code.
        \item[\rm (2)] $C_2\subseteq C_1\cap(C_1^{ 2})^\perp$.
        \item[\rm (3)] $C_1^\perp+C_1^{ 2}\subseteq C_2^\perp$.
    \end{itemize}
Moreover, if $Q(C_1, C_2)$ is a CSS-T code, then $C_2$ is self-orthogonal.
\end{theorem}
We use the following result from \cite{campsmoreno2023poset} to determine the parameters of a quantum  CSS-T code. 
\begin{corollary}\label{C:dmin2perp}
Let $(C_1, C_2)$ be a CSS-T pair. Then 
$$
\min \{\wt(C_1),\wt(C_2^\perp)\}=\wt(C_2^\perp),
$$
and the parameters of the corresponding CSS-T code are $$[[n,k_1-k_2,\ge \wt(C_2^\perp)]].$$
Moreover, if the code is nondegenerate, we have equality in the minimum distance.
\end{corollary}

\begin{definition}
A pair of codes $(C_1,C_2)$ satisfying the conditions in Theorem \ref{T:equiv-def} is called a $CSS$-T pair. 
\end{definition}

According to Theorem \ref{T:equiv-def}, given a code $C$, the intersection $C \cap (C^2)^{\perp}$
 is useful in determining a maximal 
CSS-T pair $(C_1,C_2)$ with $C_1=C$. To better understand this intersection, we note its relationship to the hull of a code.

\begin{definition} \label{def:hull}
    The hull of a code $C$ is $\hull (C):=C \cap C^{\perp}$. Given codes $C_1$ and $C_2$ of the same length, the relative hull of $C_1$ with respect to $C_2$ is $\hull_{C_2}(C_1):=C_1 \cap C_2^{\perp}$. 
\end{definition}

\begin{lemma} \label{lemma:hull_intersections}
  For a binary code $C$, 
 $$C \cap (C^2)^{\perp} = \hull (C) \cap \hull (C^2).
$$    
\end{lemma}
 
\begin{proof}
    Note that $C \subseteq C^2$, since $w \star w=w$ for all $w \in \F_2^n$. Note $C \cap (C^2)^{\perp} \subseteq C^2 \cap (C^2)^{\perp} = \hull(C^2)$. In addition,      $(C^2)^{\perp} \subseteq C^{\perp}$ and $C \cap (C^2)^{\perp} \subseteq \hull (C)$. Thus, $C \cap (C^2)^{\perp} \subseteq \hull (C) \cap \hull (C^2).$

    As $\hull (C) \cap \hull (C^2) = C \cap C^\perp \cap C^2 \cap (C^2)^\perp \subseteq C \cap (C^2)^{\perp}$, the proof is complete.
\end{proof}

We come to one of the main results of this section. Lemma \ref{lemma:hull_intersections} provides another characterization of CSS-T codes, as stated in the next result. 

\begin{theorem}\label{CSSTHull}
Given binary codes $C_1$ and $C_2$ with $C_2 \subseteq C_1$, 
$Q(C_1,C_2)$ is a CSS-T code if and only if
$$C_2 \subseteq \hull (C_1) \cap \hull (C_1^2)$$ if and only if 
    $$
    C_2 \subseteq \hull_{C_1^2}(C_1).
    $$
\end{theorem}

\begin{proof}
    This result follows immediately from Theorem \ref{T:equiv-def}, Lemma \ref{lemma:hull_intersections}, and Definition \ref{def:hull}. 
\end{proof}

Notice that Theorem \ref{CSSTHull} together with Theorem \ref{T:equiv-def} indicates that to determine a CSS-T pair $(C_1,C_2)$, one may restrict the search to self-orthogonal codes $C_2$ in the relative hull of $C_1$ with respect to its square.

\section{Puncturing}
Let $C\subset \F_2^n$ be a code and $i\in [n]$. \rmv{The shortening of $C$ in $\{i\}$, denoted by $C_{\{i\}}$, is the binary code
\begin{align*}
C_{\{i\}} = \{ & (c_1,\dots,c_{i-1},c_{i+1},\dots,c_n) \\
: &(c_1,\dots,c_{i-1},0,c_{i+1},\dots,n)\in C\}. 
\end{align*}
}
The puncturing of $C$ in $\{i\}$, denoted by $C|_i$, is the binary code
\begin{align*}
C|_i := \{ &(c_1,\dots,c_{i-1},c_{i+1},\dots,c_n) \\
: &(c_1,\dots,c_{i-1},c_i,c_{i+1},\dots,c_n)
\in C,\\
&\text{for some } c_i \in \F_2 \}. 
\end{align*}
For $S\subset[n]$, we write $C|_S$ for the successive puncturing of $C$ in the coordinates indexed by the elements in $S$. The code $C$ is degenerated on $\{i\}$ if $c_i=0$ for every $c=(c_1,\ldots,c_i,\ldots,c_n) \in C$. 

We come to one of the main results of this section. The following theorem states that if $Q(C_1,C_2)$ is a CSS-T code, then $Q(C_1|_i,C_2|_i)$ is a CSS-T code whenever the code $\hull (C_1) \cap \hull (C_1^2)$ is degenerated on $\{i\}$.

\begin{theorem}\label{03.03.24}
    Let $C_2 \subseteq C_1$ be binary codes.
    Assume that $C_2$ is degenerated on $\{i\}$. If $Q(C_1, C_2)$ is a CSS-T code, then $Q(C_1|_i, C_2|_i)$ is a CSS-T code.
    \rmv{Let $D\subseteq C\cap(C^{\ast 2})^\perp$. If $i\notin\mathrm{supp}(C\cap(C^{\ast 2})^\perp)$, then $D|_i\subseteq C|_i\cap(C|_i^{\ast 2})^\perp$.}
\end{theorem}
\begin{proof}
    Note that for any $c\in C_2 \subseteq C_1$, $c_i=0$, so $c|_i\in C_2|_i$ and $c|_i\in C_1|_i$.
    
    As $Q(C_1, C_2)$ is a CSS-T code, then $$C_2 \subseteq \hull (C_1) \cap \hull (C_1^2)$$ by Theorem~\ref{CSSTHull}. Thus,
    $c\cdot w=c|_i\cdot w|_i=0$
    for every element $w$ in $C_1^2$.
    We obtain $$C_2|_i\subseteq C_1|_i\cap(C_1|_i)^{2\perp},$$ from which we get the conclusion by Theorem~\ref{CSSTHull}.
\end{proof}

The support of $C\subset \F_2^n$ is denoted and defined by
\begin{align*}
\supp(C):=\{&i \in [n] : c_i \neq 0 \text{ for some }\\
&c=(c_1,\ldots,c_i,\ldots,c_n) \in C \}.
\end{align*}

\begin{corollary}
    Let $S$ be the complement of $\supp(C)$.
    If $Q(C_1,C_2)$ is a CSS-T code, then $(C_1|_S,C_2|_S)$ is a CSS-T code.
\end{corollary}
\begin{proof}
    This is a consequence of Theorem~\ref{03.03.24}.
\end{proof}

\section{Quasi-cyclic low-density codes}
\label{section:families}
In this section, we study quantum LDPC and LDGM codes defined by quasi-cyclic codes.
\begin{definition}
    A CSS code $Q(C_1, C_2)$ is a quantum LDPC code if $C_1$ and $C_2^\perp$ are LDPC codes, or, equivalently, if $C_1$ is an LDPC code and $C_2$ is an LDGM code. Similarly, $Q(C_1, C_2)$ is a quantum LDGM code if $C_1$ and $C_2^\perp$ are LDGM codes, or, equivalently, if $C_1$ is an LDGM code and $C_2$ is an LDPC code.
\end{definition}

\begin{remark}
    Note that a binary generator matrix for the stabilizer that defines a quantum LDPC code $Q(C_1,C_2)$ can be written as
$$
\left[ 
\begin{array}{cc}
0 & H_1 \\ 
G_2 & 0
\end{array}
\right],
$$
where $H_1$ and $G_2$ are sparse matrices.
\end{remark}

Here we use those ideas to consider LDPC codes which give rise to CSS-T pairs. As proof of concept, we make use of a code construction found in \cite{myung}, where the authors sought codes that have both efficient encoding algorithms and fast iterative decoding algorithms.


For an integer $L\geq 2$, define the matrix
$$P:=\begin{pmatrix} 0&1&0&\cdots&0\\ 0&0&1&\cdots&0\\ \vdots&\vdots&\vdots&\ddots&\vdots\\
0&0&0&\cdots&1\\
1&0&0&\cdots&0\end{pmatrix}\in\mathbb{F}_2^{L\times L}.$$

Take positive integers $m$ and $n$. Let $a_{ij}\in\mathbb{Z}_L\cup\{\infty\}$, for $0\leq i\leq m-1$ and $0\leq j\leq n-1$. Define a code $C$ by the parity check matrix
\begin{equation} \label{eq:H}
H=\begin{pmatrix}
P^{a_{00}}&P^{a_{01}}&\cdots&P^{a_{0(n-1)}}\\
\vdots&\vdots&\ddots&\vdots\\
P^{a_{(m-1)0}}&P^{a_{(m-1)1}}&\cdots&P^{a_{(m-1)(n-1)}}\end{pmatrix},\end{equation}
where $P^\infty$ denotes a square matrix of zeroes of size $L$ and $P^{a_{ij}}$ is the usual $a_{ij}$ power matrix multiplication of the matrix $P$.

Observe that if $a\in\mathbb{Z}_L$, then $P^{a}$ is a permutation matrix. Indeed, $\row_i\ P^{a}=e_{a+i\ \mathrm{mod}\ L}^T$.
Note that the code $C$ is a quasi-cyclic LDPC (QC-LDPC) code. The weight of each row is $\leq n$ and the weight of each column is $\leq m$. Moreover, by knowing $a$, we can recover immediately $P^a$. Thus, we can store $H$ with a smaller base matrix
$$M_H=\begin{pmatrix}
a_{00}&\cdots&a_{0(n-1)}\\
\vdots&\ddots&\vdots\\
a_{(m-1)0}&\cdots&a_{(m-1)(n-1)}\end{pmatrix}\in(\mathbb{Z}_L\cup\{\infty\})^{m\times n}.$$

We aim  to describe the square of $C$ in terms of the entries of $M_H$. To that end, for $a\in(\mathbb{Z}_L\cup\{\infty\})^n$, let $R(a)\in\mathbb{F}_2^{Ln}$ be defined as
$$[R(a)]_i=\begin{cases}
1& \text{if } i=jL+a_j\text{ for some }j\in\{0,\cdots,n-1\}\\& \text{such that } a_j\neq \infty\\
0&\text{otherwise}.\end{cases}$$
Making use of this definition, the following Lemma gives a natural connection between $M_H$ and the rows of $H$.

\begin{lemma}\label{03.03.24.L}
    Let $C$ be a quasi-cyclic LDPC code with shift $L$ and assume $a\in(\mathbb{Z}_L\cup\{\infty\})^n$ is such that $R(a)\in C$. Extend the sum over $\mathbb{Z}_L$ to $\mathbb{Z}_L\cup\{\infty\}$ by taking $x+\infty=\infty$ for any $x\in \mathbb{Z}_L\cup\{\infty\}$ and take $\mathbf{1}=(1,\ldots,1)\in\mathbb{Z}_L^n$. Then
    $$R(a+j\mathbf{1})\in C,\ \forall j\in\mathbb{Z}_L.$$
\end{lemma}

\begin{proof}
    Observe that 
    $$R(a_i+j)=
    \begin{cases}
        R(a_i+j) & \text{if } a_i\neq \infty\\
        R(\infty)& \text{if } a_i=\infty.
    \end{cases}$$
    In either case, $R(a_i+j)=R(a_i) \cdot P$ and thus $R(a + j \mathbf{1})\in C$.
\end{proof}

Consider the operation $\ast:(\mathbb{Z}_L\cup\{\infty\})^2\rightarrow\mathbb{Z}_L\cup\{\infty\}$ defined by 
$$
a \ast b = 
\begin{cases}
a & \text{if } a=b \\ \infty & \text{otherwise}.
\end{cases}
$$ 
The operation extends naturally component-wise to $(\mathbb{Z}_L\cup\{\infty\})^n$. Moreover, there is a relationship between $\star$ and $\ast$, as shown in the next result. 

\begin{proposition}\label{03.03.24.P}
Given $a, b \in (\mathbb{Z}_L\cup\{\infty\})^n$,
$$R(a)\star R(b)=R(a\ast b).$$
\end{proposition}

\begin{proof}
    Let $a=(a_0,\ldots,a_{n-1})$, $b=(b_0,\ldots,b_{n-1})$ and $0\leq i\leq n-2$. We will focus on the entries of $R(a)$ indexed by $iL+(i+1)L-1$, meaning $R(a_i)$. Observe that $R(a_i)=e_{a_i}$, the $a_i^{th}$ standard basis vector in $\mathbb{F}_2^L$. Thus, $$R(a_i)\star R(b_i)=e_{a_i}\star e_{b_i}
    =\begin{cases}
        0 & \text{if } a_i\neq b_i \\
        e_{a_i} & \text{otherwise}.
    \end{cases}
           $$Thus, 
    $$R(a_i)\star R(b_i)=R(a_i\ast b_i).$$
    Since $R(a)$ is the concatenation of $R(a_i)$, we have the conclusion.
\end{proof}

\begin{proposition} Given the code $C^\perp$ defined by parity check matrix $H$ as in Equation (\ref{eq:H}) and its corresponding matrix $M_H$ with rows $A_1,\ldots, A_m$, its square 
    $(C^\perp)^2$ is generated by a matrix $H' \in \F_2^{mL \times nL}$ such that $M_{H'}  \in(\mathbb{Z}_L\cup\{\infty\})^{m\times n}$ has rows $A_i\ast (A_j+h\mathbf{1})$ for any $0\leq h\leq L-1.$
\end{proposition}

\begin{proof}
    We know that $C^\perp$ is generated by $R(A_i+j\mathbf{1})$ with $0\leq j\leq L-1$. Thus, by Proposition \ref{03.03.24.P}, the square is generated by
    $$R(A_{i_1}+{j_1}\mathbf{1})\star R(A_{i_2}+{j_2}\mathbf{1})=R((A_{i_1}+{j_1}\mathbf{1})\ast (A_{i_2}+{j_2}\mathbf{1})).$$

    The conclusion will follow if $w=(A_{i_1}+{j_1}\mathbf{1})\ast (A_{i_2}+{j_2}\mathbf{1})=A_{i_1}\ast(A_{i_2}+h\mathbf{1})+h'\mathbf{1}$ for some $h,h'\in\mathbf{Z}_L$. We claim that
    $$w=A_{i_1}\ast(A_{i_2}+(j_2-j_1)\mathbf{1})+j_1\mathbf{1}.$$

    Observe that if $(A_{i_1})_\nu+j_1=(A_{i_2})_\nu+j_2$ then either $(A_{i_1})_\nu=(A_{i_2})_\nu=\infty$ or $(A_{i_1})_\nu=(A_{i_2})_\nu+(j_2-j_1)$. In the first case, $w_\nu=\infty$ and in the second case $w_\nu=(A_{i_2})_\nu+j_2$. In any case, we have
    $$w_\nu=(A_{i_1})_\nu\ast ((A_{i_2})_\nu+(j_2-j_1))+j_1.$$
    
    On the other hand, if $(A_{i_1})_\nu+j_1\neq (A_{i_2})_\nu+j_2$, then necessarily $w_\nu=\infty$. If $(A_{i_1})_\nu\neq\infty$, then $(A_{i_1})_\nu\neq (A_{i_2})_\nu+(j_2-j_1)$ and
    $$w_\nu=(A_{i_1})_\nu\ast ((A_{i_2})_\nu+(j_2-j_1))+j_1.$$

    Similarly for $(A_{i_2})_\nu\neq \infty$ and then we have the conclusion.

    Since $R(w)\in (C^\perp)^{\ast 2}$, by Lemma \ref{03.03.24.L}, we have $R(w+h\mathbf{1})\in (C^\perp)^2$ for any $h$ and thus, we can store any of them to build $M_{H'}$, from where we have the conclusion by taking $h=-j_1\mathrm{mod}\ L$.
\end{proof}

\section{Code} \label{section:examples}
In this section, we present Magma~\cite{magma} code based on Theorem~\ref{03.03.24} and quasi-cyclic codes to find quantum LDPC or LDGM CSS-T codes computationally. We also describe the algorithms in case one wishes to use a different software, for instance Macaulay2~\cite{Mac2} along with the coding theory package~\cite{cod_package}.

We start by giving the algorithms to compute the Schur product between two matrices, the square (respect the Schur product) of a matrix, and the $H$ matrix given in Eq.~\ref{eq:H}.
\begin{algorithm}[ht!]
    \begin{algorithmic}[1]
        \Function{Pointwise}{$A$, $B$} \Comment{Returns the pointwise matrix between same-size matrices $A$ and $B$}
            \State $n \gets \left| \{\text{columns of $A$}\}\right|$
            \State $m \gets \left| \{\text{rows of $A$}\} \right|$
            \State $C \gets 0_{n \times m}$
            \For{$j\in [n]$, $i \in [m]$}
                \State $C[i,j] \gets A[i,j] * B[i,j]$
            \EndFor
            \State \textbf{return} $C$
        \EndFunction
    \end{algorithmic}
\end{algorithm}

\begin{algorithm}[ht!]
    \begin{algorithmic}[1]
        \Function{Square}{$A$} \Comment{Returns the (Schur) square of $A$}
            \State $n \gets \left| \{\text{columns of $A$}\}\right|$
            \State $m \gets \left| \{\text{rows of $A$}\} \right|$
            \State $C \gets 0_{m^2 \times n}$
            \State $\ell \gets 1$
            \For{$i, j \in [m]$}
                \State $C[\ell] \gets$ Pointwise(Row $i$ of $A$, Row $j$ of $A$)
                \State $\ell \gets \ell + 1$
            \EndFor
            \State \textbf{return} $C$
        \EndFunction
    \end{algorithmic}
\end{algorithm}

We now present the Magma code that can be used to find LDPC or LDGM quantum CSS-T codes.
The following function returns the Schur product between same-size matrices $A$ and $B$.
\begin{verbatim}
function Pointwise(A,B)
 return Matrix([[A[i,j]*B[i,j] : 
 j in [1..Ncols(A)]] : 
 i in [1..Nrows(A)]]);
end function;
\end{verbatim}

The following function returns the square of the matrix $A$ (in terms of the Schur product).
\begin{verbatim}
function Square(A)
 return Matrix([Pointwise(RowSubmatrix
 (A, i, 1),RowSubmatrix(A, j, 1))
 : i,j in [1..Nrows(A)]]);
end function;
\end{verbatim}

The following function returns the quasi-cyclic low-density matrix $H$ defined by the integer $L$ and the matrix $A$ (see Eq.~\ref{eq:H}) that can be used as generator matrix for a QC-LDGM code or parity-check matrix for a QC-LDPC code.


\begin{verbatim}
function QCLD(L,A)
 P:=ZeroMatrix(FiniteField(2), L, L);
 InsertBlock(~P, IdentityMatrix 
 (FiniteField(2), L-1), 1, 2);
 P[L,1]:=FiniteField(2)!1;
 
 H:=ZeroMatrix(FiniteField(2), 
 Nrows(A)*L, Ncols(A)*L);
 for i in [0..Nrows(A)-1] do
  for j in [0..Ncols(A)-1] do
  InsertBlock(~H,P^A[i+1,j+1], 
  i*L+1, j*L+1);
  end for;
 end for;
 return H;
end function;
\end{verbatim}
\begin{example}
We will use the previous functions to generate a CSS-T code. Specifically, we use the QCLD function to provide a sparse generator matrix.

\begin{verbatim}
L:=4;
A:= Matrix(IntegerRing(), 
2, 4, [3,1,2,1,  3,3,2,3] );

G:=QCLD(L,A);
C1:=LinearCode(G);
C2:=C1 meet Dual(LinearCode
(Square(GeneratorMatrix(C1))));
C1;
Dual(C2);
\end{verbatim}
Then $C_1$ is a $[16, 6, 4]$ binary code with generator matrix\newline
$G_1 = $
\[
\begin{bmatrix}
1&0&0&0&0&0&1&0&0&0&0&1&0&0&1&0\\
0&1&0&0&0&0&0&1&1&0&0&0&0&0&0&1\\
0&0&1&0&0&0&1&0&0&1&0&0&0&0&1&0\\
0&0&0&1&0&0&0&1&0&0&1&0&0&0&0&1\\
0&0&0&0&1&0&1&0&0&0&0&0&1&0&1&0\\
0&0&0&0&0&1&0&1&0&0&0&0&0&1&0&1
    \end{bmatrix}.
\]
Moreover, $$C_2 = \hull (C_1) \cap \hull (C_1^2) =C_1$$ and $C_2^\perp$ is a $[16, 10, 2]$ binary code. So, by Corollary~\ref{C:dmin2perp}  and Theorem~\ref{CSSTHull}, the quantum code $Q(C_1,C_2)$ is a
$[[16,0,\ge 2]]$ CSS-T code.
\end{example}

As an additional example of the techniques introduced in this paper, we provide the following.

\begin{example}\rm
    Let $C_1$ and $C_2$ be defined by the generator matrices
    $$\begin{bmatrix}
         1&1&1&1&1&1&1&1&1&1&1&1&1&1&1&1\\
         1&1&1&1&1&1&1&1&0&0&0&0&0&0&0&0\\
         1&1&1&1&0&0&0&0&1&1&1&1&0&0&0&0\\
         1&1&0&0&1&1&0&0&1&1&0&0&1&1&0&0\\
         1&0&1&0&1&0&1&0&1&0&1&0&1&0&1&0\\
         \end{bmatrix}$$
    and 
    $$\begin{bmatrix}
         1&1&1&1&1&1&1&1&0&0&0&0&0&0&0&0\\
         1&1&1&1&0&0&0&0&1&1&1&1&0&0&0&0\\
         1&1&0&0&1&1&0&0&1&1&0&0&1&1&0&0\\
         1&0&1&0&1&0&1&0&1&0&1&0&1&0&1&0\\
         \end{bmatrix},$$
     respectively. We can check that $(C_1,C_2)$ is a CSS-T using Theorem~\ref{CSSTHull} and \cite{magma}. As the code $C_2$ is degenerated with respect to the last column, we can puncture that column to obtain that $((C_1)|_{15},(C_2)_{15})$ is also a CSS-T pair by Theorem~\ref{03.03.24}. The minimum distance of $((C_2)|_{15})^\perp$ is 3 and then we get a CSS-T code with parameters $[[15,1,3]]$.
\end{example}

\section{Conclusion} \label{section:conclusion}
In this paper, we provided a characterization of CSS-T codes using the relative hull of a code with respect to its square. We proved that under certain conditions, we can puncture the component codes of a CSS-T pair to obtain another CSS-T pair. We considered the use of quasi-cyclic codes to design LDPC and LDGM quantum CSS-T codes computationally. Toy examples were given as a proof of concept, demonstrating a possible step towards obtaining CSS-T codes using the characterization.

\bibliographystyle{IEEEtran}
\bibliography{bibliography}

\end{document}